\documentclass[12pt, 
]{article}

\usepackage{booktabs}

\usepackage[utf8]{inputenc}
\usepackage{amssymb}
\usepackage{amsmath}

\usepackage{geometry}
\newtheorem{theorem}{Theorem}

\newtheorem{corollary}{Corollary}
\newtheorem{definition}{Definition}
\newtheorem{example}{Example}
\newtheorem{lemma}{Lemma}
\newtheorem{remark}{Remark}

\newenvironment{proof}[1][Proof]{\emph{#1.} }{\  \hfill $\square $ \vspace{5 pt}}

\usepackage{natbib}
\usepackage{amssymb}
\usepackage[dvips]{graphics}
\usepackage{amsmath}

\usepackage{mathpazo}
\usepackage{enumerate}

\usepackage{amsfonts}

\usepackage{amssymb}

\usepackage{enumerate}

\usepackage{mathrsfs}

\usepackage[usenames]{color}
\usepackage{cancel}

\usepackage{pgf,tikz}

\usetikzlibrary{arrows.meta,positioning,fit,calc,backgrounds}

\usetikzlibrary{decorations.pathreplacing,angles,quotes}
\usetikzlibrary{arrows.meta}
\usetikzlibrary{arrows, decorations.markings}
\tikzset{myptr/.style={decoration={markings,mark=at position 1 with %
       {\arrow[scale=2,>=stealth]{>}}},postaction={decorate}}}

\usepackage[pagebackref]{hyperref}
\hypersetup{
    colorlinks,
    citecolor=blue,
    filecolor=black,
    linkcolor=blue,
    urlcolor=blue
}

\usepackage{pgf,tikz}
\usetikzlibrary{arrows}

\usepackage{fancyhdr}

\usepackage{soul}

\usepackage{emptypage}

\usepackage[T1]{fontenc}
\DeclareFontFamily{T1}{calligra}{}
\DeclareFontShape{T1}{calligra}{m}{n}{<->s*[1.44]callig15}{}
\DeclareMathAlphabet\mathcalligra   {T1}{calligra} {m} {n}

\usepackage{multirow}
\usepackage{array}
\usepackage{mathtools}
\usepackage{arydshln}
\usepackage{threeparttable}
\usepackage{booktabs,caption}

\usepackage{ifthen}
\newboolean{showcomments}
\setboolean{showcomments}{true}

\newcommand{\pablo}[1]{  \ifthenelse{\boolean{showcomments}}
{\textcolor{green!50!black}{(T: #1)}}{}}
\newcommand{\marcelo}[1]{\ifthenelse{\boolean{showcomments}}
{\textcolor{red}{(M: #1)}}{}}
\newcommand{\agustin}[1]{  \ifthenelse{\boolean{showcomments}}
{\textcolor{blue!50!black}{(T: #1)}}{}}

\usepackage{orcidlink}
\usepackage{mathtools}

\usepackage{circledsteps}
\usepackage{pstricks}
\usepackage{authblk}

\begin{document}

\title{
Lattice operations for the stable set in many-to-many markets via re-equilibration dynamics\thanks{%
We are grateful to Nadia Guiñazú, Jordi Massó, Bumin Yenmez, and especially William Thomson for their detailed comments. We also thank participants in the GATE Saint-Étienne Seminar at Université Jean Monnet,  the 8th International Workshop on Matching Under Preferences MATCH-UP 2026 at NYU Paris, and the 25th Annual Conference of the Society for the
Advancement of Economic Theory SAET 2026 at the Instituto de Matemática Pura e Aplicada (IMPA).
E-mails: \href{mailto:agustinbonifacio@gmail.com}{\texttt{agustinbonifacio@gmail.com}} (A. G. Bonifacio),
\href{mailto:noemjuarez@gmail.com}{\texttt{noemjuarez@gmail.com}} (N. Juarez),
\href{mailto:pbmanasero@email.unsl.edu.ar}{\texttt{pbmanasero@email.unsl.edu.ar}} (P. B. Manasero).}}

\author[1,2]{Agustín G. Bonifacio  
\orcidlink{0000-0003-2239-8673}}
\author[1]{Noelia Juarez \orcidlink{0000-0002-1610-2031}}
\author[1]{ Paola B. Manasero \orcidlink{0000-0003-0384-499X}}
\affil[1]{Instituto de Matem\'{a}tica Aplicada San Luis, 
Universidad Nacional de San Luis and CONICET, San Luis, Argentina.}
\affil[2]{GATE Lyon--St-Étienne, Université Jean Monnet, Saint-\'Etienne, France.}

\date{
}

\maketitle

\begin{abstract}
We compute the lattice operations for the (pairwise) stable set in many-to-many matching markets when only path-independence on agents' choice functions is imposed.  
To do this, we first show that the sets of firm-quasi-stable and worker-quasi-stable many-to-many matchings form lattices.
Then, we construct Tarski operators on these lattices whose fixed points coincide with the set of stable matchings, and show that iterating these operators from suitable quasi-stable matchings yields the lattice operations in the stable set.
These operators resemble lay-off and vacancy chain dynamics, respectively. 

\bigskip


\medskip


\end{abstract}

\section{Introduction}

In many-to-many matching markets there are two sides, typically workers and firms, and the main solution concept is that of a ``pairwise stable'' allocation. An allocation is pairwise stable if no worker and firm who are not matched to each other would both prefer to be matched together, adjusting their current sets of partners.\footnote{In many-to-many matching there are several notions of stability. Here, we only study \emph{pairwise stability}. For a comprehensive presentation of stability and other core-related notions in the many-to-many setting, see \cite{echenique2004theory}.} The lattice structure of the set of pairwise stable allocations---hereafter simply \emph{stable} allocations---is a fundamental tool in this 
theory: it underlies results on conflicts of interest between sides, alignment of interests within sides,  and is helpful---among other things---to obtain the full set of stable matchings \citep{bonifacio2022cycles}. In this paper, we compute the lattice operations in many-to-many markets under the sole requirement that agents’ choice functions be path-independent. This axiom, originally due to \cite{plott1973path}, ensures that the outcome of choice is invariant to the order in which alternatives are considered.

Since \cite{aizerman1981general} it is well known that path-independence is equivalent to substitutability, once the mild condition of consistency is invoked as well.\footnote{In the context of matching with contracts,  consistency is called ``irrelevance of rejected alternatives''\citep[see][]{aygun2013matching}.}
Substitutability, introduced in the matching literature by \cite{kelso1982job}, states that agents are still chosen when the set of alternatives shrinks and they are still available, so no complementarities among agents prevail.  
Hence, while substitutability has traditionally been imposed directly on preferences 
to ensure the existence of stable matchings, the path-independence axiom captures exactly the same restriction at the level of choice behavior. Moreover, taking choice rules instead of preferences as primitives of matching markets has many advantages \citep[see][for a detailed discussion]{chambers2017choice}. Recent work provides a utility-based foundation for this behavioral formulation. \cite{hafalir2026market} show that ordinal concavity of an institution’s distributional objective is sufficient for their distribution-conscious choice rule to be path-independent. More generally, \cite{yokote2023rationalizing} establish that a choice rule is path-independent if and only if it can be rationalized by an ordinally concave utility function.

In the one-to-one model with strict preferences, it is well known that the set of stable matchings forms a lattice under the unanimous order for workers. Given two stable matchings, a new matching is obtained by assigning to each worker her preferred partner among the two. Conway’s insight, reproduced in \cite{knuth1976marriages} and formalized in \cite{roth1992two}, is that this pointwise assignment is itself stable and coincides with the least stable matching that dominates both (i.e., their join).  
The greatest lower bound (i.e., the meet) is constructed by assigning to each worker the worst partner among the two. A dual exercise can be performed to compute the lattice operations under the unanimous order for firms.

In the many-to-one model with path-independent choice functions, the simple constructions of the one-to-one case no longer hold. A coarser partial order---Blair's order---is needed to compare matchings and obtain the lattice structure. Moreover, either more stringent conditions on the model are necessary or a much more cumbersome theoretical structure is to be imposed in order to compute the lattice operations. Following the first approach,  the join and meet between stable matchings have been obtained invoking, in addition to path-independence, the property of  ``separability with quota'' \citep{marti2001lattice}. 
Following the second approach, \cite{echenique2004core}  compute the lattice operations at the expense of performing their analysis in the realm of pre-matchings, 
whose economic interpretation is difficult to grasp. Their approach is related to the fixed-point literature on stable matchings initiated by \cite{adachi2000characterization}, who derives the stable marriage theorem from Tarski’s fixed-point theorem, and further developed by \cite{fleiner2003fixed}, who provides a broader lattice-theoretic framework for bipartite stable matching problems. \cite{echenique2004core} define a Tarski operator within the set of pre-matchings and show that the formulas for the join and meet obtained by \cite{marti2001lattice} define only pre-matchings when path-independence alone is required, but the fixed points obtained by iterating their operator from such pre-matchings deliver the desired join and meet stable matchings.\footnote{We frame our analysis in terms of path-independent choice functions, while the cited papers take preferences as primitives and therefore refer to substitutability rather than path-independence. }

Here, inspired by \cite{echenique2004core}'s approach, we tackle the many-to-many setting from a different perspective. Instead of working with pre-matchings, we propose to work with quasi-stable matchings. A quasi-stable matching may contain blocking pairs, but only of a kind that does not disrupt existing relationships on one side of the market. Therefore, such matchings arise in two variants: firm-quasi-stable and worker-quasi-stable. 

In many-to-one scenarios, these sets of quasi-stable matchings admit a lattice structure. \cite{wu2018lattice} study firm-quasi-stable matchings\footnote{Actually, \cite{wu2018lattice} study the slightly stronger notion of ``worker-envy-free'' matching. Firm-quasi-stable matchings in many-to-one models generalize the one-to-one notion 
of ``simple'' matchings \citep{sotomayor1996non}  and have been studied by \cite{cantala2004restabilizing,cantala2011agreement} among others.} in  markets with responsive preferences (a much more restrictive requirement than path-independent choice functions) and show 
that firm-quasi-stable matchings form a lattice with respect to the unanimous order for workers, while \cite{bonifacio2022lattice} study worker-quasi-stable matchings under path-independence and show their lattice structure with respect to Blair's order for firms.  Such structures enable the definition of Tarski operators over these matchings. The fixed points of these operators correspond to fully stable matchings, providing a systematic method to analyze how a market can transition from quasi-stability to stability through re-equilibration processes.


While many-to-one models with path-independent choice functions are well understood, extending the analysis to many-to-many settings remains nontrivial. Our first challenge is to determine whether the lattice structure of quasi-stable matchings extends to this more general framework. This question is relevant not only in its own right, but also because quasi-stable matchings may provide a useful intermediate structure for computing the lattice operations in the stable set. Although \cite{blair1988lattice} establishes that the set of stable matchings in many-to-many markets forms a lattice under path-independence, his construction of the lattice operations lacks a transparent economic interpretation. This motivates us to look for an alternative route through quasi-stability. Given two stable matchings, we first construct a new matching by assigning to each firm the subset selected by its choice function from the union of its partners under the two matchings. The need for a subsequent re-equilibration is already apparent in many-to-one markets: as \cite{li2013note}  observed, under path-independence alone, the matching obtained through this natural choice-based operation need not be stable. In the many-to-many setting, an analogous issue arises when the construction is performed from the workers' side. We show, however, that these constructions produce quasi-stable matchings. They can therefore be re-equilibrated until stability is restored, providing a natural route for recovering the lattice operations in the stable set.

We implement this approach by introducing natural notions of firm-quasi-stability and worker-quasi-stability based on the preservation of existing relationships, and by showing that the corresponding sets of quasi-stable matchings form lattices under the corresponding Blair orders.\footnote{A general definition of firm-quasi-stable matching in the setting of many-to-many matchings with contracts can be found in \cite{yang2025existence}.} Building on the insight of \cite{echenique2004core}, but using quasi-stable matchings rather than pre-matchings, we then define Tarski operators on these lattices whose fixed points coincide with stable matchings. Iterating the appropriate operator from the quasi-stable matching constructed above yields the join of the two original stable matchings: the firm-side construction, which produces a worker-quasi-stable matching, yields their join under Blair's order for firms, while the symmetric worker-side construction yields their join under  Blair's order for workers.

Once the join of two stable matchings has been computed with respect to one 
side of the market, the duality between Blair's orders for firms and workers
under path-indepen\-dence \citep{blair1988lattice} implies that this join coincides with the 
meet with respect to the other side. Hence, both lattice operations are obtained 
for both sides of the market.

Tarski operators in matching models admit a natural interpretation as re-equilibration processes in decentralized markets. This perspective is rooted in the vacancy-chain dynamics of \cite{blum1997vacancy}, where sequences of job reallocations driven by unfilled positions restore stability in labor markets. In the present framework and related work, this idea is formalized through an isotone operator whose fixed points coincide with stable matchings. Economically, each iteration corresponds to a round of decentralized adjustments: in worker-quasi-stable environments, the process takes the form of layoff chains, where unemployed workers apply to preferred firms, firms update their workforce, and the resulting dismissals propagate further adjustments; dually, in firm-quasi-stable environments, the presence of unfilled positions gives rise to vacancy-chain–like dynamics, in which reallocations are driven by the sequential filling of vacancies. These processes generate monotone paths in the lattice of quasi-stable matchings and converge to stability, providing a unified and operational interpretation of equilibrium as the endpoint of endogenous market correction.

    \cite{blair1988lattice} provides a method for computing lattice operations without reference  to quasi-stability or Tarski operators. We show that Blair's insight is in fact equivalent to our fixed-point method applied to the lattices of many-to-many quasi-stable matchings.\footnote{We shed light on this in Example \ref{example Blair}.} Thus, our approach provides a sound economic foundation for Blair's original construct.

    The remainder of the paper is organized as follows. In Section \ref{sect preliminares}, we present our model and some preliminaries. Section \ref{section quasi y Tarski} is devoted to presenting the concepts of worker-quasi-stability and firm-quasi-stability, and their respective Tarski operators. Section \ref{section lattice operations} contains our main result: the computation of the lattice operations for the stable set.  Finally, some conclusions are gathered in Section \ref{sect conclusions}.

\section{Preliminaries}\label{sect preliminares}

A \textbf{(matching) market} consists of two finite disjoint sets, the set of firms $F$ and the set of workers $W$. Let $A=F \cup W$ and, for each $i \in A$, let $$A_i=\begin{cases}
F \text{ \ if  \ } i \in W, \\ W \text{  \ if  \ }i \in F.    
\end{cases}$$
Each agent $i \in A$ has a choice function $C_i: 2^{A_i} \longrightarrow 2^{A_i}$ that satisfies \textbf{path-independence}, which says that
\begin{equation}\label{propiedad de choice}
C_i\left(S\cup S'\right)=C_i\left(C_i\left(S\right)\cup S'\right)
\end{equation} for each pair of subsets $S$ and $S'$ of $A_i$.

Under a regularity condition called ``consistency'',\footnote{\textbf{Consistency:} $C_i(S')=C_i(S)$ whenever $C_i(S)\subseteq S' \subseteq S \subseteq A_i.$} path-independence is equivalent to \textbf{substitutability}, that requires for each pair of subsets $S$ and $S'$ of $A_i$ that $s \in S' \subseteq S$ together with $s \in C_i(S)$ imply $s \in C_i(S')$.\footnote{Actually, a choice function is path-independent if and only if it is substitutable and satisfies consistency.}

A  market is denoted by $\boldsymbol{(C_F,C_W)}$, where $C_F$ is the profile of choice functions for all firms and $C_W$ is the profile of choice functions for all workers. A matching associates firms and workers. Formally,

\begin{definition}\label{def matching}
A \textbf{matching} is a function $\mu: A \longrightarrow 2^{A}$ such that, for each $w\in W$ and  each $f\in F$,
\begin{enumerate}[(i)]
\item $\mu(w)\subseteq F$, 
\item $\mu(f)\subseteq W$, and 
\item $w\in \mu(f)$ if and only if $f \in \mu(w)$. 
\end{enumerate}
\end{definition}
Let $\boldsymbol{\mathcal{M}}$ denote the set of all matchings for market $(C_F, C_W)$.

Let $i \in A$ and $\mu \in \mathcal{M}$. Agent $i$ is \textbf{matched} if $\mu(i) \neq \emptyset$, otherwise  $i$ is \textbf{unmatched}. Matching $\mu$ is \textbf{blocked by agent $\boldsymbol{i}$} if $\mu(i)\neq C_i(\mu(i))$. Matching $\mu$ is \textbf{individually rational} if it is not blocked by any agent. Matching $\mu$ is \textbf{blocked by a firm-worker pair $\boldsymbol{(f,w)}$} if (i) $f \notin \mu(w)$, (ii) $w \in C_f(\mu(f)\cup \{w\})$, and (iii) $f \in C_w(\mu(w)\cup \{f\})$. Matching $\mu$ is \textbf{(pairwise) stable} if it is individually rational and it is not blocked by any firm-worker pair. Let $\boldsymbol{\mathcal{S}}$ denote the set of all stable matchings of market $(C_F,C_W)$.

Within the set of all matchings, we can define a partial order from each side of the market as follows.\footnote{Given a set $\mathcal{X}$, a \emph{partial order} $\geq$ over $\mathcal{X}$ is a reflexive, antisymmetric, and transitive binary relation. If this is the case, sometimes we refer to the \emph{partially ordered set} $(\mathcal{X}, \geq)$.} Let $\mu, \mu'\in \mathcal{M}$. We say that \textbf{$\boldsymbol{\mu$ is (Blair) preferred to $\mu'}$ by the firms}, and write $\boldsymbol{\mu \succeq_F \mu'}$, if $$C_f(\mu(f)\cup \mu'(f))=\mu(f) \text{ for each }f \in F.$$ Similarly, we say that \textbf{$\boldsymbol{\mu$ is (Blair) preferred to $\mu'}$ by the workers}, and write $\boldsymbol{\mu \succeq_W \mu'}$, if $$C_w(\mu(w)\cup \mu'(w))=\mu(w)\text{ for each }w \in W.$$
An important fact about the set of stable matchings is that it is a lattice with respect to both partial orders $\succeq_F$ and $\succeq_W$.\footnote{Given a partially ordered set $(\mathcal{X},\geq)$, and two elements $x,y \in \mathcal{X}$, an element $z \in \mathcal{X}$ is an \emph{upper bound of $x$ and $y$} if $z \geq x$ and $z \geq y$. An element $ x  \vee  y \in \mathcal{X}$ is the \emph{join} (or \emph{supremum}) \emph{of $x$ and $y$} if and only if (i) $x \vee  y$ is an upper bound of $x$ and $y$, and (ii) $t \geq x  \vee  y$ for each upper bound $t$ of $x$ and $y$. The definitions of \emph{lower bound} and \emph{meet} (or \emph{infimum}) \emph{of $x$ and $y$}, denoted $x  \wedge  y$, are dual and we omit them. Furthermore, $(\mathcal{X},\geq)$ is a \emph{lattice} if $x  \vee  y$ and $x  \wedge  y$ exist for each pair $x,y \in \mathcal{X}$.} Moreover, $\succeq_F$ and $\succeq_W$ are \emph{dual} orders within $\mathcal{S}$, meaning that if $\mu, \mu' \in \mathcal{S}$ then $\mu \succeq_F \mu'$ if and only if $\mu' \succeq_W \mu$ \citep[see][]{blair1988lattice}. 

In order to compute the lattice operations for such lattices, we will consider two enlargements of the set of stable matchings: the set of worker-quasi-stable matchings and the set of firm-quasi-stable matchings. In each one, a Tarski operator is to be used to compute the join for the partial order that endows it with a lattice structure.

\section{Tarski operators for quasi-stable matchings}\label{section quasi y Tarski} 

Quasi-stable matchings are a relaxation of the classical notion of stability in matching markets. Unlike fully stable matchings, quasi-stable matchings may contain certain blocking pairs, but only those that do not disrupt existing assignments on a protected side of the market. 

In this section, we extend the notions of quasi-stability to many-to-many markets, show that the sets of worker-quasi-stable and firm-quasi-stable matchings admit lattice structures under path-independence, and define and analyze Tarski operators on these lattices.

For the notion of worker-quasi-stability, 
blocking pairs are allowed as long as they do not compromise the already existing relations for the workers in the matching. For a matching $\mu \in \mathcal{M}$ and a worker $w \in W$, consider the set of firms that currently employ $w$ or that would be willing to add $w$ to their workforce. Formally, define
$$
F_w^\mu=\{f \in F : w \in C_f(\mu(f)\cup\{w\})\}.$$
\vspace{-20 pt}
\begin{definition}\label{def worker-quasi many-to-many}
    Matching $\mu$ is \textbf{worker-quasi-stable} if it is individually rational and, for each $w \in W$ and each $T \subseteq F_w^\mu$, we have $$\mu(w) \subseteq C_w(\mu(w)\cup T).$$ 
\end{definition}
Denote by $\boldsymbol{\mathcal{Q}^W}$ the set of all worker-quasi-stable matchings for market $(C_F,C_W)$. Notice that, since the empty matching (where every agent is unmatched) belongs to this set, $\mathcal{Q}^W\neq \emptyset$. 

\begin{remark} \em
    The symmetric definition of \emph{firm-quasi-stability} can be obtained by interchanging the roles of sets $W$ and $F$ in Definition \ref{def worker-quasi many-to-many}. The set of all firm-quasi-stable matchings is denoted by $\boldsymbol{\mathcal{Q}^F}$. \hfill $\triangle$
\end{remark}


In a many-to-one model with substitutable and $q$-separable preferences,\footnote{The combination of substitutability and $q$-separability is stronger than path-independence.} given two stable matchings \cite{marti2001lattice} define a new matching that assigns to each firm its most-preferred subset of workers among those assigned to it in the two original matchings. This  new matching is also stable in that setting, and yields the desired lattice structure. The generalization of that construction to our many-to-many environment, though not always stable, will prove useful. We present it next. 

Given $\mu,\mu' \in \mathcal{Q}^W,$ we define matching $\boldsymbol{\lambda_{\mu,\mu'}}$ as follows:
\begin{enumerate}[(i)]
\item for each $f\in F,$ $\lambda_{\mu,\mu'} (f)=C_f\left(\mu (f)\cup \mu ^{\prime }(f)\right),$
\item for each $w \in W,$ $\lambda_{\mu,\mu'} (w)=\{f\in F : w\in \lambda_{\mu,\mu'}(f)\}.$
\end{enumerate}

In \cite{bonifacio2022lattice} it is shown, in a many-to-one setting, that matching $\lambda_{\mu,\mu'}$ is well-defined, i.e., it is a worker-quasi-stable matching. Moreover, it is the join between $\mu$ and $\mu'$ with respect to $\succeq_F$ within the worker-quasi-stable set, implying that  this set is a join-semilattice\footnote{A partially ordered set $(\mathcal{X}, \geq)$ is a \emph{join-semilattice} if $x \ \vee \ y$ exists for each pair $x,y \in \mathcal{X}$.} and, moreover, a lattice. Next, we generalize these facts to our many-to-many model.

\begin{theorem}\label{prop facts worker-quasi many} (Facts about worker-quasi-stable matchings)
    \begin{enumerate}[(i)]
        \item Let $\mu, \mu' \in \mathcal{Q}^W$. Then, $\lambda_{\mu, \mu'} \in \mathcal{Q}^W$. Furthermore, $\lambda_{\mu, \mu'}$ is the join of $\mu$ and $\mu'$ with respect to $\succeq_F$ in $\mathcal{Q}^W$.
        \item $(\mathcal{Q}^W, \succeq_F)$ is a lattice. 
    \end{enumerate}
\end{theorem}
\begin{proof}
    See Subsection \ref{proof facts worker-quasi many} in the Appendix. 
\end{proof}

\begin{remark}\label{facts firm-quasi} \em
    The symmetric version of the previous result states that if, given two firm-quasi-stable matchings $\mu$ and $\mu'$, we define 
    matching $\boldsymbol{\gamma_{\mu,\mu'}}$ as follows:
\begin{enumerate}[(i)]
\item for each $w\in W,$ $\gamma_{\mu,\mu'} (w)=C_w\left(\mu (w)\cup \mu ^{\prime }(w)\right),$
\item for each $f \in F,$ $\gamma_{\mu,\mu'} (f)=\{w\in W : f \in \gamma_{\mu,\mu'}(w)\},$
\end{enumerate} then matching $\gamma_{\mu,\mu'}$ is firm-quasi-stable and, moreover, it is the join between $\mu$ and $\mu'$ with respect to $\succeq_W$. It also states that $(\mathcal{Q}^F, \succeq_W)$ is a lattice. \hfill $\triangle$
\end{remark}

Since market $(C_F,C_W)$ is symmetric in $F$ and $W$, in what follows we construct a many-to-many Tarski operator for worker-quasi-stable matchings. Of course, the dual construction for firm-quasi-stable matchings is straightforward.

Given $\mu \in \mathcal{Q}^W$ and $f \in F$, let 
\begin{equation}\label{b tilde de f}
    B_f^\mu=\mu(f) \cup \left\{w \in W : f \in C_w(F_w^\mu)\right\}.
\end{equation}
Set $B_f^\mu$ consists of all workers who are either matched with firm $f$ under $\mu$, or who, together with $f$, form a blocking pair of $\mu$.

We are now in a position to define the many-to-many version of the operator for matchings in $\mathcal{Q}^W$.

\begin{definition}{(Tarski operator for worker-quasi-stable matchings)}\label{def operator worker many-to-many} For each $\mu \in \mathcal{Q}^W$, operator $\mathcal{T}^F:\mathcal{Q}^W \longrightarrow \mathcal{Q}^W$ assigns 
 \begin{enumerate}[(i)]
    \item for each $f \in F$, $\mathcal{T}^F[\mu](f)=C_f(B_f^\mu)$
    \item for each $w \in W$, $\mathcal{T}^F[\mu](w)=\{f \in F : w \in \mathcal{T}^F[\mu](f)\}.$
\end{enumerate}   
\end{definition}

 \begin{remark} \em
     The dual definition of \emph{Tarski operator for firm-quasi-stable matchings}, $\mathcal{T}^W$, can be obtained by interchanging the roles of sets $W$ and $F$ in Definition \ref{def operator worker many-to-many}. 
    It is easy to see that, in a many-to-one model where firms have responsive preferences instead of path-independent choice functions, operator $\mathcal{T}^W$ specializes to operator $T$ defined in \cite{wu2018lattice}. \hfill $\triangle$
 \end{remark}

 Operator $\mathcal{T}^F$ is (i) well-defined and Pareto-improving for the firms, (ii) isotone,  and (iii) has as its set of fixed points the stable set. Formally,

\begin{theorem}\label{results operator worker many-to-many}
    For operator $\mathcal{T}^F: \mathcal{Q}^W \longrightarrow \mathcal{Q}^W$ we have:

    \begin{enumerate}[(i)]
        \item For each $\mu \in \mathcal{Q}^W$, $\mathcal{T}^F[\mu] \in \mathcal{Q}^W$ and $\mathcal{T}^F[\mu] \succeq_F \mu$.
        \item If $\mu, \mu' \in \mathcal{Q}^W$ and $\mu \succeq_F \mu'$, then $\mathcal{T}^F[\mu] \succeq_F \mathcal{T}^F[\mu']$. 
        \item $\mathcal{T}^F[\mu]=\mu$ if and only if $\mu \in \mathcal{S}$.
    \end{enumerate}
\end{theorem}
\begin{proof}
    See Subsection \ref{proof results worker many-to-many} in the Appendix. 
\end{proof}

\begin{remark} \em
Since $\mathcal{T}^F$ is isotone by Theorem \ref{results operator worker many-to-many} (ii), the set of its fixed points is a non-empty lattice with respect to $\succeq_F$ according to Tarski's Fixed Point Theorem. Since this set is exactly $\mathcal{S}$ by Theorem \ref{results operator worker many-to-many} (iii), as a byproduct, we obtain an alternative proof that $\mathcal{S}$ is non-empty and has a lattice structure. \hfill $\triangle$ 
\end{remark}

To analyze the re-equilibration dynamics implicit in our Tarski operators, we focus on the one defined on the set of worker-quasi-stable matchings and interpret it as a decentralized  process in a labor market with frictions. Take as initial condition a quasi-stable allocation arising after a shock---e.g., entry of new workers, firm downsizing, or a reorganization of teams---so that some profitable firm--worker deviations exist, but current worker-side relationships are temporarily protected. One iteration of the operator represents a round of local adjustments: workers involved in blocking opportunities apply to their most preferred admissible firms; each firm then re-optimizes by selecting its preferred subset from the pool of incumbents and applicants, given its choice function; as a result, some previously employed workers may be displaced and become available for the next round. Because agents can maintain multiple relationships, these adjustments do not simply create or destroy matches, but reallocate workers across firms, and the induced displacements propagate as a layoff chain across the market. The operator thus maps the current configuration into the outcome of one such decentralized round, generating a monotone path in the lattice of worker-quasi-stable matchings. Fixed points are precisely those allocations where no further profitable local adjustments arise---i.e., stable matchings---so convergence captures the market's endogenous return to equilibrium after a shock \citep{bonifacio2022lattice}. A symmetric interpretation applies to the operator defined on firm-quasi-stable matchings, where the adjustment is driven by vacancy-chain dynamics.

\section{Lattice operations}\label{section lattice operations}

Given two stable matchings $\mu$ and $\mu'$, the natural candidates to be the join between them with respect to partial orders $\succeq_F$ and $\succeq_W$ in the stable set are $\lambda_{\mu,\mu'}$ and $\gamma_{\mu,\mu'}$, respectively. However, in general, $\lambda_{\mu,\mu'}$ and $\gamma_{\mu,\mu'}$ are only \emph{quasi-stable} matchings (see Theorem \ref{prop facts worker-quasi many} (i), Remark \ref{facts firm-quasi}, and Example \ref{example Li} next). The following result shows that $\lambda_{\mu, \mu'}$ can be re-equilibrated by  iteratively applying operator $\mathcal{T}^F$ to obtain the join between $\mu$ and $\mu'$ with respect to $\succeq_F$ within $\mathcal{S}$. Similarly, $\gamma_{\mu, \mu'}$ can be re-equilibrated by  iteratively applying operator $\mathcal{T}^W$ to obtain the join between $\mu$ and $\mu'$ with respect to $\succeq_W$ within $\mathcal{S}$. 

Given $\mu \in \mathcal{Q}^F$,  denote by $\boldsymbol{\mathcal{F}^W(\mu)}$ the fixed point of $\mathcal{T}^W$ starting from $\mu$. Similarly, given $\mu \in \mathcal{Q}^W$,  denote by $\boldsymbol{\mathcal{F}^F(\mu)}$ the fixed point of $\mathcal{T}^F$ starting from $\mu$. 

\begin{figure}[ht]
\begin{center}
\begin{tikzpicture}


\fill[blue!3] (-6,-3) -- (6,-3) -- (6,5) -- (-6,5) -- (-6,-3);










\draw[blue, thick] (0,0) ellipse (4cm and 2cm);
\draw[blue, thick] (0,1) ellipse (5cm and 3.5cm);





\draw[
blue] (0,3.5) to[out=10,in=95] (1.25,3);
\draw[
dashed, blue] (1.25,3) to[out=-80,in=80] (1.25,2.25);
\draw[
blue] (1.25,2.25) to[out=-100,in=-10] (0,1);



\draw[fill=blue] (-2,-0.75) circle[radius=1.5 pt] node[left]{\textcolor{blue}{$\mu$}};

\draw[fill=blue] (2,-0.75) circle[radius=1.5 pt] node[right]{\textcolor{blue}{$\mu'$}};

\draw[fill=blue] (0,3.5) circle[radius=1.5 pt] node[above]{\textcolor{blue}{$\lambda_{\mu,\mu'}$}};

\draw[fill=blue] (0,1) circle[radius=1.5 pt] node[above]{\textcolor{blue}{$\mu \curlyvee_F \mu'$}};

\draw[fill=blue] (1.25,3) circle[radius=1.5 pt] node[above right]{\textcolor{blue}{$\mathcal{T}^F(\lambda_{\mu,\mu'})$}};

\draw[fill=blue] (1.25,2.25) circle[radius=1.5 pt] node[right]{\textcolor{blue}{$\mathcal{T}^{F^{(k)}}(\lambda_{\mu,\mu'})$}};

\draw[fill=blue] (-3.25,2) 
node[below]{\textcolor{blue}{\large{$\mathcal{S}$}}};

\draw[fill=blue] (-4.25,4) 
node[below]{\textcolor{blue}{\large{$\mathcal{Q}^W$}}};






\end{tikzpicture}
\caption[computing the join]{\label{figurajoin}\small{\textsf{\textbf{Computing the join w.r.t. $\boldsymbol{\succeq_F}$ between two stable matchings}. Given two stable matchings $\mu$ and $\mu'$, first we compute worker-quasi-stable matching $\lambda_{\mu,\mu'}$. Then, we re-equilibrate it by iteratively applying operator $\mathcal{T}^F$ until obtaining its stable fixed point  $\mu \curlyvee_F \mu'$. }}}
\end{center}
\end{figure}


\begin{theorem}\label{theorem join and meet many-to-many}
Let $\mu, \mu' \in \mathcal{S}$. Then, 
    \begin{enumerate}[(i)]
        \item $ \mu  \curlyvee_F 
  \mu'=\mathcal{F}^F(\lambda_{\mu, \mu'})$, and 
        \item $ \mu  \curlywedge_F  \mu'=\mathcal{F}^W(\gamma_{\mu, \mu'})$.
    \end{enumerate}
\end{theorem}
\begin{proof}
    See Subsection \ref{proof of theorem join and meet many-to-many} in the Appendix.
\end{proof}

\noindent Part (i) of the previous theorem is illustrated in Figure \ref{figurajoin}. Part (ii) of the previous theorem is implied by the duality between orders $\succeq_F$ and $\succeq_W$, that also implies the following corollary.


\begin{corollary}
    Let $\mu, \mu' \in \mathcal{S}$. Then, 
    \begin{enumerate}[(i)]
        \item $\mu  \curlyvee_W   \mu' =\mathcal{F}^W(\gamma_{\mu, \mu'})$, and 
        \item $ \mu  \curlywedge_W 
  \mu'=\mathcal{F}^F(\lambda_{\mu, \mu'})$.
    \end{enumerate}
\end{corollary}



In what follows, we present two illustrative examples that highlight different aspects 
of our results. 

The first one, taken directly from \cite{li2013note}, shows that $\lambda_{\mu,\mu'}$ is not always stable, contrary to the claim made by \cite{roth1985conflict}.\footnote{\cite{blair1988lattice} also observes that $\lambda_{\mu,\mu'}$ need not be stable. See Footnote 11 in \cite{li2013note}.} Moreover, this example highlights that, even in the many-to-one model, stability may fail when no restriction other than path-independence is  imposed on preferences. 

\begin{example}\label{example Li}
Let $(C_F,P_W)$ be a many-to-one market with $F=\{f_1,f_2,f_3, f_4, f_5\}$ and $W=\{w_1, w_2,$ $w_3, w_4, w_5, w_6\}$.
Preferences of the agents are given in Table \ref{tabla ejemplo Li}. Agents' choice functions are derived from these preferences in the standard way.\footnote{For example, $C_{f_1}(W)=\{w_4\}$ and $C_{f_4}(\{w_1,w_2,w_3, w_4, w_6\})=\{w_4,w_6\}$.}

\begin{table}[h!]
\centering
\setlength{\tabcolsep}{6 pt} 
\renewcommand{\arraystretch}{2} 
\begin{tabular}{|c|c c c c c c:c| c c c c c c|}
\hline
$P_{f_1}$ & $\boxed{\overline{w_4}^\star}$ & $\underline{w_1}$ & ${w_5}^\dagger$ & $\cdots$  & \Circled{ $\emptyset$ } & $\cdots$ & $P_{w_1}$ & \Circled{ $\overline{f_2}^\dagger$ } & $\underline{f_1}$ & ${f_3}^\star$ & $\cdots$ & $\boxed{\emptyset}$ & $\cdots$ \\

$P_{f_2}$ & $\boxed{\underline{w_2}^\star}$ & \Circled{ $\overline{w_1,w_3}^\dagger$ } & $\cdots$ &  &  &  & $P_{w_2}$ & \Circled{ $\overline{f_3}^\dagger$ } & $\boxed{\underline{f_2}^\star}$ & $\cdots$ &  &  & \\

$P_{f_3}$ & ${w_1}^\star$ & $\boxed{\underline{w_3}}$ & \Circled{ $\overline{w_2}^\dagger$ } & $\cdots$ &  &  & $P_{w_3}$ & \Circled{ $\overline{f_2}^\dagger$ } & $\boxed{\underline{f_3}}$ & $\cdots$ & $\emptyset^\star$ & $\cdots$ & \\

$P_{f_4}$ & $\boxed{\overline{w_5}^\star}$ & \Circled{ $\underline{w_4,w_6}^\dagger$ } & $\cdots$ &  &  &  & $P_{w_4}$ & 
\Circled{ $\underline{f_4}^\dagger$ } & $\boxed{\overline{f_1}^\star}$ & $\cdots$ &  &  & \\

$P_{f_5}$ & $\boxed{\overline{w_6}^\star}$ & \Circled{ $\underline{w_5}$ } & $\cdots$ & $\emptyset^\dagger$ & $\cdots$ &  & $P_{w_5}$ & $f_1^\dagger$ & \Circled{ $\underline{f_5}$ } & $\boxed{\overline{f_4}^\star}$ & $\cdots$ &  & \\

          &  &  &  &  &  &  & $P_{w_6}$ & \Circled{ $\underline{f_4}^\dagger$ } & $\boxed{\overline{f_5}^\star}$ & $\cdots$ &  &  & \\

\hline 
\end{tabular}
\caption{Preference profile for Example \ref{example Li}.}
\label{tabla ejemplo Li}
\end{table}


\noindent Let  
$$\underline{\mu}=\begin{pmatrix}
    f_1 & f_2 & f_3 & f_4 & f_5 \\
    \{w_1\} & \{w_2\} & \{w_3\} & \{w_4,w_6\} & \{w_5\}\\
\end{pmatrix}$$
and 
$$\overline{\mu}=\begin{pmatrix}
    f_1 & f_2 & f_3 & f_4 & f_5 \\
    \{w_4\} & \{w_1,w_3\} & \{w_2\} & \{w_5\} & \{w_6\}\\
\end{pmatrix}.$$ Then,
$$\boxed{\mu}=\lambda_{\underline{\mu},\overline{\mu}}=\begin{pmatrix}
    f_1 & f_2 & f_3 & f_4 & f_5 & \emptyset\\
    \{w_4\} & \{w_2\} & \{w_3\} & \{w_5\} & \{w_6\} & \{w_1\}\\
\end{pmatrix}$$ and 
$$\text{\Circled{ $\mu$ }}=\gamma_{\underline{\mu},\overline{\mu}}=\begin{pmatrix}
    f_1 & f_2 & f_3 & f_4 & f_5 \\
    \emptyset & \{w_1,w_3\} & \{w_2\} & \{w_4,w_6\} & \{w_5\}\\
\end{pmatrix}.$$
Notice that both $\underline{\mu}$ and $\overline{\mu}$ are stable matchings.\footnote{Partners in matching $\underline{\mu}$ are depicted in Table 1 in the same manner. For example, to show that $f_4$ is matched to $w_4$ and $w_6$ under $\underline{\mu}$ we write $\underline{w_4,w_6}$ in the preference of $f_4$, and $\underline{f_4}$ in the preferences of $w_4$ and $w_6$. Partners in matchings $\overline{\mu}$, $\boxed{\mu}$, $\Circled{\mu}$, $\mu^\star$, and $\mu^\dagger$ are depicted similarly.}  We know, by \cite{bonifacio2022lattice}, that $\boxed{\mu}$ is the join of $\underline{\mu}$ and $\overline{\mu}$ with respect to $\succeq_F$ in the  worker-quasi-stable set. However, as \cite{li2013note} points out, $\boxed{\mu}$ is not stable, since $(f_3,w_1)$ blocks it. Similarly, \Circled{ $\mu$ } is the join of $\underline{\mu}$ and $\overline{\mu}$ with respect to $\succeq_W$ in the  firm-quasi-stable set by Remark \ref{facts firm-quasi}. However, \Circled{ $\mu$ } is not stable either, since $(f_1,w_5)$ blocks it. Applying the respective Tarski operator once, we get
$$
\mu^\star=\mathcal{T}^F\left[\boxed{\mu}\right]=\begin{pmatrix}
    f_1 & f_2 & f_3 & f_4 & f_5 & \emptyset\\
    \{w_4\} & \{w_2\} & \{w_1\} & \{w_5\} & \{w_6\} & \{w_3\}\\
\end{pmatrix}$$ and 
$$
\mu^\dagger=\mathcal{T}^W\left[\text{\Circled{ $\mu$ }}\right]=\begin{pmatrix}
    f_1 & f_2 & f_3 & f_4 & f_5 \\
    \{w_5\} & \{w_1,w_3\} & \{w_2\} & \{w_4,w_6\} & \emptyset\\
\end{pmatrix}.$$
It is readily seen that $\mu^\star$ is the firm-optimal matching, so it is stable. By Theorem \ref{theorem join and meet many-to-many} (i), it is $\mu  \curlyvee_F   \mu'$. Similarly, $\mu^\dagger$ is the worker-optimal matching, so it is stable and by Theorem \ref{theorem join and meet many-to-many} (ii), it is $\mu  \curlywedge_F   \mu'$. \hfill $\Diamond$
\end{example}

The second example \citep[Example 2 in][]{blair1988lattice} 
shows that computing the lattice operations between two stable matchings can be substantially more involved. 
The example is also used to show that the lattice is not necessarily distributive. Here we revisit this example with the aim of showing that, by working with the approach based on worker-quasi-stable matchings and with the Tarski operator we have defined, the binary operations between stable matchings can be obtained much more easily.

\begin{table}
\centering
\vspace{-30 pt}
\setlength{\tabcolsep}{6 pt} 
\renewcommand{\arraystretch}{2} 
\begin{tabular}{|c|c c c c c c c c c c|}

\hline

 $P_{f_1}$ & \Circled{ ${w_1}^\star$ } & $\boxed{\overline{\underline{w_2,w_3, w_4}}}$ & $\cdots$ &  &  &  &  &  &  &  \\
 $P_{f_2}$ & ${w_2}^\star$ & $w_8$ & \Circled{ $\boxed{\overline{\underline{w_1,w_5}}}$ } & $\cdots$ &  &  &  &  &  &  \\
 $P_{f_3}$ & ${w_3}^\star$ & \Circled{ $\boxed{\overline{w_9}}$ } & $\underline{w_1,w_6}$ & $\cdots$ &  &  &  &  &  &  \\
 $P_{f_4}$ & ${w_4}^\star$ & \Circled{ $\boxed{\underline{w_{10}}}$ } & $\overline{w_1,w_7}$ & $\cdots$ &  &  &  &  &  &  \\
 $P_{f_5}$ & ${w_5}^\star$ & \Circled{ $\boxed{\overline{\underline{w_8}}}$ } & $\cdots$ &  &  &  &  &  &  &  \\
 $P_{f_6}$ & \Circled{ $\boxed{\overline{w_6}^\star}$ } & $\underline{w_9}$ & $\cdots$ &  &  &  &  &  &  &  \\
 $P_{f_7}$ & \Circled{ $\boxed{\underline{w_7}^\star}$ } & $\overline{w_{10}}$ & $\cdots$ &  &  &  &  &  &  &  \\

\hdashline

 $P_{w_1}$ & $f_2,f_3, f_4$ & $\underline{f_2,f_3}$ & $f_3,f_4$ & $\overline{f_2,f_4}$ & \Circled{ $f_1,f_2$ } & $f_1,f_3$ & $f_1,f_4$ & ${f_1}^\star$ & $\boxed{f_2}$  & $\cdots$ \\
 $P_{w_2}$ & $\boxed{\overline{\underline{f_1}}}$   & ${f_2}^\star$ & \Circled{ $\emptyset$ } &  &  &  &  &  &  &  \\
 $P_{w_3}$ & $\boxed{\overline{\underline{f_1}}}$ & ${f_3}^\star$ & \Circled{ $\emptyset$ } &  &  &  &  &  &  &  \\
 $P_{w_4}$ & $\boxed{\overline{\underline{f_1}}}$ & ${f_4}^\star$ & \Circled{ $\emptyset$ } &  &  &  &  &  &  &  \\
 $P_{w_5}$ & \Circled{ $\boxed{\overline{\underline{f_2}}}$ } & ${f_5}^\star$ & $\cdots$ &  &  &  &  &  &  &  \\
 $P_{w_6}$ & $\underline{f_3}$ & \Circled{ $\boxed{\overline{f_6}^\star}$ } & $\cdots$ &  &  &  &  &  &  &  \\
 $P_{w_7}$ & $\overline{f_4}$ & \Circled{ $\boxed{\underline{f_7}^\star}$ } & $\cdots$ &  &  &  &  &  &  &  \\
 $P_{w_8}$ & \Circled{ $\boxed{\overline{\underline{f_5}}}$ } & $f_2$ & $\emptyset^\star$ &  &  &  &  &  &  &  \\
 $P_{w_9}$ & $\underline{f_6}$ & \Circled{ $\boxed{\overline{f_3}}$ } & $\emptyset^\star$  &  &  &  &  &  &  &  \\
 $P_{w_{10}}$ & $\overline{f_7}$ & \Circled{ $\boxed{\underline{f_4}}$ } & $\emptyset^\star$ &  &  &  &  &  &  &  \\
 \hline
\end{tabular}
\caption{Preference profile for Example \ref{example Blair}.}
\label{tabla ejemplo Blair}
\end{table}

\begin{example}\label{example Blair}
Let $F=\{f_1, \ldots, f_7\}$ and $W=\{w_1, \ldots, w_{10}\}$. Consider the path-independent many-to-many market induced by the preferences in Table \ref{tabla ejemplo Blair}.  It is readily seen that matchings $$\underline{\mu}=\begin{pmatrix}
    f_1 & f_2 & f_3 & f_4 & f_5 & f_6 & f_7\\
    \{w_2,w_3,w_4\} & \{w_1,w_5\} & \{w_1,w_6\} & \{w_{10}\} & \{w_8\} & \{w_9\} & \{w_7\} \\
\end{pmatrix}$$ and 
$$\overline{\mu}=\begin{pmatrix}
    f_1 & f_2 & f_3 & f_4 & f_5 & f_6 & f_7\\
    \{w_2,w_3,w_4\} & \{w_1,w_5\} & \{w_9\} & \{w_1, w_7\} & \{w_8\} & \{w_6\} & \{w_{10}\} \\
\end{pmatrix}$$ are stable. If we compute matching   $\lambda_{\underline{\mu},\overline{\mu}}$ 
we obtain
 $$\boxed{\mu}=\lambda_{\underline{\mu},\overline{\mu}}= \begin{pmatrix}
    f_1 & f_2 & f_3 & f_4 & f_5 & f_6 & f_7\\
    \{w_2,w_3,w_4\} & \{w_1,w_5\} & \{w_9\} & \{w_{10}\} & \{w_8\} & \{w_6\} & \{w_7\} \\
\end{pmatrix}.$$ Matching $\boxed{\mu}$ is worker-quasi-stable but not stable since, for instance, $(f_1,w_1)$ blocks $\boxed{\mu}$. Let us now use the Tarski operator to compute $\mathcal{T}^F\left[\boxed{\mu}\right]$. 
We get

$$\text{\Circled{ $\mu$ }}=\mathcal{T}^F\left[\boxed{\mu}\right]=\begin{pmatrix}
    f_1 & f_2 & f_3 & f_4 & f_5 & f_6 & f_7 & \emptyset\\
    \{w_1\} & \{w_1,w_5\} & \{w_9\} & \{w_{10}\} & \{w_8\} & \{w_6\} & \{w_7\} & \{w_2,w_3,w_4\} \\
\end{pmatrix}.$$ Again,  matching \Circled{ $\mu$ } is worker-quasi-stable but not stable since, for example, $(f_2,w_2)$ blocks \Circled{ $\mu$ }. A second application of the operator generates matching $\mathcal{T}^{F^{(2)}}\left[\boxed{\mu}\right]$: 

$$\mu^\star=\mathcal{T}^{F^{(2)}}\left[\boxed{\mu}\right]=\begin{pmatrix}
    f_1 & f_2 & f_3 & f_4 & f_5 & f_6 & f_7 & \emptyset\\
    \{w_1\} & \{w_2\} & \{w_3\} & \{w_4\} & \{w_5\} & \{w_6\} & \{w_7\} & \{w_8,w_9,w_{10}\} \\
\end{pmatrix}.$$
Matching $\mu^\star$ is stable. In fact, it is the firm-optimal matching. Therefore, $\mu^\star= \underline{\mu}  \curlyvee_F  \overline{\mu}$. Blair’s discussion implicitly anticipates the re-equilibration process described above. In his explanation of this example, he says: ``the join between $\underline{\mu}$ and $\overline{\mu}$ cannot have worker $w_1$ hired by firms $f_3$ or $f_4$. Thus, worker $w_1$ will want to work for firm $f_1$. This means that firm $f_1$ will not want workers $w_2, w_3$, and $w_4$, who will want to work for firms $f_2, f_3,$ and $f_4$, respectively. Since workers $w_5, w_6,$ and $w_7$ have no alternative employment, they will work for $f_5, f_6,$ and $f_7$, respectively. Thus, the join has every firm get its first choice''. 
\hfill $\Diamond$
\end{example}

\section{Conclusions}\label{sect conclusions}

In this paper, we compute the lattice operations for the stable set 
when only path-indepen\-dence on agents' choice functions is imposed. A few final remarks are in order.

If besides being path-independent, choice functions are assumed to satisfy the ``law of aggregate demand'', that says that when a firm chooses from an expanded set it hires at least as many workers as before,\footnote{\textbf{Law of aggregate demand:} $S'\subseteq S\subseteq W$ implies $|C_f(S')|\leq |C_f(S)|$ \citep[see][]{alkan2002class,hatfield2005matching}. \cite{alkan2002class} calls this property ``cardinal monotonicity''. } then $\lambda_{\mu,\mu'}$ and $\gamma_{\mu,\mu'}$ are stable matchings. Moreover, the same additional properties obtained in \cite{bonifacio2022lattice} for many-to-one quasi-stable matchings  extend to the present many-to-many framework. In particular, (i) the fixed point reached by iterating the Tarski operator from a quasi-stable matching coincides with the join of that matching and the corresponding side-optimal stable matching; (ii) joins between quasi-stable and stable matchings preserve stability; (iii) quasi-stable matchings dominating the side-optimal stable matching are themselves stable; and (iv) the size of agents’ assigned sets is bounded by their size in stable matchings. Analogous results had previously been obtained by \cite{wu2018lattice}  in the more restrictive setting of many-to-one matching with responsive preferences.

It is also worth mentioning that, under the law of aggregate demand, \cite{alkan2002class} derives the lattice operations and shows that the lattice of stable matchings is distributive, a property that no longer holds under path-independence alone \citep[see also][for a different proof]{li2014new}. Under the same additional assumption, \cite{manasero2021binary} obtains the lattice operations by means of a morphism connecting many-to-many and many-to-one matching models, lifting the corresponding operations through this bridge.\footnote{The lattice operations can also be obtained in a more straightforward fashion if, in addition to path-independence, choice functions are assumed to be ``quota-filling'' \citep{alkan2001preferences}.}

In many-to-many matching with contracts, when choice functions are path-independent on one side and preferences are responsive on the other, \cite{bonifacio2024lattice} show the lattice structure of the related set of ``envy-free'' matchings. \cite{bonifacio2024lattice} also study re-equilibration by means of a Tarski operator. Firm-quasi-stability generalizes their notion of envy-freeness. In the absence of contracts, our operator in Definition \ref{def operator worker many-to-many} generalizes the one presented in \cite{bonifacio2024lattice}. In fact, our results can be extended to matching markets with contracts. 

In order to prove that the sets of  many-to-many (worker or firm) quasi-stable matchings are lattices, we follow \cite{wu2018lattice} and \cite{bonifacio2022lattice}: we show that such sets are join-semilattices with a minimum. However, how to directly compute the meet between any pair of quasi-stable matchings remains an open question, even invoking additional restrictions such as the law of aggregate demand.

\bibliographystyle{ecta}
\bibliography{bibliolattice}

\appendix

\section{Proofs}\label{appendix general}

\subsection{Proof of Theorem \ref{prop facts worker-quasi many}}\label{proof facts worker-quasi many}

First, we present a useful result about the behavior of set $F_w^\mu$. 

\begin{lemma}\label{lema Blair} (Facts about $F_w^\mu$)
    \begin{enumerate}[(i)]
        \item Let $\mu \in \mathcal{M}$ be such that $C_f(\mu(f))=\mu(f)$ for each $f \in F$.  
        Then, $\mu(w) \subseteq F_w^\mu$ for each $w \in W$.
        \item Let $\mu, \mu' \in \mathcal{M}$ and $S_f,S'_f \subseteq W$ be such that $S_f \subseteq S_f'$, 
        $\mu(f)=C_f(S_f)$, and $\mu'(f)=C_f(S_f')$ for each $f \in F$. Then, $F_w^{\mu'} \subseteq F_w^\mu$ for each $w \in W.$  
    \end{enumerate}
\end{lemma}
\begin{proof}
    (i) Let $\mu \in \mathcal{M}$ be such that $C_f(\mu(f))=\mu(f)$ for each $f \in F$, and let $w \in W$. If $f \in \mu(w)$, then $w \in \mu(f)$. Thus, $\mu(f)=\mu(f)\cup \{w\}$ and $w \in \mu(f)=C_f(\mu(f) \cup \{w\})$, so $f \in F_w^\mu$.

    \vspace{5 pt}

    \noindent (ii) Let $\mu, \mu' \in \mathcal{M}$ and $S_f,S'_f \subseteq W$ for each $f \in F$ be as stated in the Lemma, and let $w \in W$. If $f \in F_w^{\mu'}$, by \eqref{propiedad de choice} we have $$w \in C_f(\mu'(f)\cup \{w\})=C_f\left(C_f(S_f')\cup \{w\}\right)=C_f(S'_f \cup \{w\}).$$ Since $S_f \subseteq S_f'$, by substitutability and  path-independence we obtain  $$w \in C_f(S_f \cup \{w\})=C_f\left(C_f(S_f) \cup \{w\}\right)=C_f(\mu(f) \cup \{w\}).$$ Hence, $f \in F_w^\mu.$ 
\end{proof}

\begin{remark} \em Part (i) of Lemma \ref{lema Blair} is a restatement of Proposition 4.7 in \cite{blair1988lattice}, whereas part (ii) is a restatement of Proposition 4.9 therein. We include their proofs for completeness.  \hfill $\triangle$
\end{remark}

\noindent \emph{Proof of Theorem \ref{prop facts worker-quasi many}.} Let $\mu, \mu' \in \mathcal{Q}^W$. For short, let $\lambda \equiv \lambda_{\mu,\mu'}$.

\smallskip
    \noindent (i)  First, we show that $\lambda$ is individually rational. Let $f \in F$. Then, by \eqref{propiedad de choice}, $$C_f\left(\lambda(f)\right)=C_f\left(C_f\left(\mu (f)\cup \mu'(f)\right)\right)=C_f\left(\mu (f)\cup \mu'(f)\right)=\lambda(f),$$ and $C_f\left(\lambda(f)\right)=\lambda(f).$

\medskip

\noindent \textbf{Claim: $\boldsymbol{\lambda(w)\subseteq C_w(F_w^{\lambda}) \text{ for each }w \in W}$.} Let $w \in W$. Define, for each $f \in F$,  $S_f=\mu(f)$ and $S_f'=\mu(f) \cup \mu'(f)$. Then, $S_f \subseteq S_f'$, $\mu(f)=C_f(S_f)$ and $\lambda(f)=C_f(S_f').$ Therefore, by Lemma \ref{lema Blair} (ii), it follows that $F_w^{\lambda} \subseteq F_w^\mu$. This last fact, together with $\mu \in \mathcal{Q}^W$, $\mu(w) \subseteq F_w^\mu$ (by Lemma \ref{lema Blair} (i)), and substitutability, imply $$\mu(w) \cap F_w^{\lambda} \subseteq C_w(F_w^\mu) \cap F_w^{\lambda} \subseteq C_w(F_w^{\lambda}),$$ so $\mu(w) \cap F_w^{\lambda}  \subseteq C_w(F_w^{\lambda})$ and, since by Lemma \ref{lema Blair} (i) we have $\lambda(w) \subseteq F_w^{\lambda},$ it follows that 
\begin{equation*}
    \mu(w) \cap \lambda(w)  \subseteq C_w(F_w^{\lambda}).
\end{equation*} 
In an analogous way we can prove that $\mu'(w) \cap \lambda(w)  \subseteq C_w(F_w^{\lambda}).$ Finally, since $\lambda(w)\subseteq \mu(w) \cup \mu'(w)$, we get $\lambda(w) \subseteq C_w(F_w^{\lambda}).$ This proves the claim. 

\medskip

Now, let $w \in W$. By the Claim, $\lambda(w) \subseteq C_w(F_w^{\lambda}).$ Then, since $\lambda(w) \subseteq F_w^{\lambda}$ by Lemma \ref{lema Blair} (i), substitutability implies  $\lambda(w) \subseteq C_w(\lambda(w)).$ Moreover, as $C_w(\lambda(w)) \subseteq \lambda(w)$, we get $C_w(\lambda(w))=\lambda(w).$ Hence, $\lambda$ is individually rational. 

To finish the proof that $\lambda$ is worker-quasi-stable, let $w \in W$ and $T \subseteq F_w^\lambda$. By the Claim,  $\lambda(w) \subseteq C_w(F_w^{\lambda}).$ Then, since $\lambda(w) \subseteq F_w^{\lambda}$ by Lemma \ref{lema Blair} (i), substitutability implies $\lambda(w) \subseteq C_w(\lambda(w) \cup T).$

To see that $\lambda$ is the join of $\mu$ and $\mu'$ with respect to $\succeq_F$ in $\mathcal{Q}^W$, first notice that, by \eqref{propiedad de choice}, $$C_f({\lambda}(f)\cup \mu(f))=C_f(C_f(\mu(f)\cup \mu'(f))\cup \mu(f))=C_f(\mu(f)\cup \mu'(f))=\lambda(f)$$ for each $f \in F$, so $\lambda \succeq_F \mu$. Similarly, $\lambda\succeq_F \mu'$, and thus $\lambda$ is an upper bound of $\mu$ and $\mu'$. Let $\nu \in \mathcal{Q}^W$ be another upper bound of $\mu$ and $\mu'$. Then, by path-independence,   $$\nu(f)=C_f(\nu(f)\cup \mu(f))=C_f(C_f(\nu(f)\cup \mu'(f))\cup \mu(f))=C_f(\nu(f)\cup \mu(f)\cup \mu'(f))=$$ $$=C_f(\nu(f)\cup C_f(\mu(f)\cup \mu'(f)))=C_f(\nu(f)\cup \lambda(f))$$ for each $f \in F$, so $\nu \succeq_F \lambda.$

\smallskip

\noindent (ii) Consider the empty matching $\mu_\emptyset$ where every agent is unmatched. Clearly, $\mu_\emptyset \in \mathcal{Q}^W$. Moreover, $\mu \succeq_F \mu_\emptyset$ for each $\mu \in \mathcal{Q}^W$, so $(\mathcal{Q}^W,\succeq_F)$ has a minimum.  By Theorem \ref{prop facts worker-quasi many} (i), $(\mathcal{Q}^W,\succeq_F)$ is a join-semilattice. Any join-semilattice with a minimum is a lattice \citep[see, for example,][]{stanley2011enumerative}. \hfill $\square$

\subsection{Proof of Theorem \ref{results operator worker many-to-many}}\label{proof results worker many-to-many}

\noindent (i) Let $\mu \in \mathcal{Q}^W$ and let $f \in F.$ Then, $$C_f(\mathcal{T}^F[\mu](f))=C_f(C_f(B_f^\mu ))=C_f(B_f^\mu)=\mathcal{T}^F[\mu](f).$$
    Let $w \in W$. Using a similar reasoning as the one applied to matching $\lambda$ to obtain the Claim in part (i) of the proof of Theorem \ref{prop facts worker-quasi many}, but this time to matching $\mathcal{T}^F[\mu]$, it follows that
    \begin{equation}\label{Blair 1}
        \mathcal{T}^F[\mu](w) \subseteq C_w\left(F_w^{\mathcal{T}^F[\mu]}\right).
    \end{equation}
    Since, by Lemma \ref{lema Blair} (i),  $\mathcal{T}^F[\mu](w) \subseteq F_w^{\mathcal{T}^F[\mu]}$,  substitutability and  \eqref{Blair 1} imply $\mathcal{T}^F[\mu](w) \subseteq C_w\left(\mathcal{T}^F[\mu](w)\right)$. As $C_w\left(\mathcal{T}^F[\mu](w)\right) \subseteq \mathcal{T}^F[\mu](w)$ as well, we have $C_w\left(\mathcal{T}^F[\mu](w)\right)=\mathcal{T}^F[\mu](w)$. This proves that $\mathcal{T}^F[\mu]$ is individually rational. Next, let $T \subseteq F_w^{\mathcal{T}^F[\mu]}$. The fact that $\mathcal{T}^F[\mu](w) \subseteq F_w^{\mathcal{T}^F[\mu]}$, \eqref{Blair 1}, and substitutability, imply $\mathcal{T}^F[\mu](w) \subseteq C_w\left(\mathcal{T}^F[\mu](w) \cup T\right)$. Thus,  $\mathcal{T}^F[\mu] \in \mathcal{Q}^W$.

    Moreover, $\mu(f)\subseteq B_f^\mu $ and path-independence   imply that  $$C_f(\mathcal{T}^F[\mu](f) \cup \mu(f))=C_f\left(C_f( B_f^\mu) \cup \mu(f)\right)=C_f(\mu(f)\cup B_f^\mu)=C_f(B_f^\mu)=\mathcal{T}^F[\mu](f).$$ As $f$ is arbitrary,  $\mathcal{T}^F[\mu] \succeq_F \mu.$

    \bigskip

    \noindent (ii) Let  $\mu, \mu' \in \mathcal{Q}^W$ be such that $\mu \succeq_F \mu'$. Assume $\mathcal{T}^F[\mu] \succeq_F \mathcal{T}^F[\mu']$ does not hold. Then, there is $f \in F$ such that 
    \begin{equation}\label{yes 1}
        \mathcal{T}^F[\mu](f) \neq C_{f}\left(\mathcal{T}^F[\mu](f) \cup \mathcal{T}^F[\mu'](f)\right). 
    \end{equation}
    By path-independence,  $C_f\left(\mathcal{T}^F[\mu](f) \cup \mathcal{T}^F[\mu'](f)\right)=C_f\left( C_f( B_f^\mu) \cup C_f( B_f^{\mu'})\right)=C_f\left( B_f^\mu\cup B_f^{\mu'}\right)$. Thus, \eqref{yes 1} becomes
    \begin{equation}\label{yes 2}
       C_f( B_f^\mu)\neq C_f\left( B_f^\mu\cup B_f^{\mu'} \right). 
    \end{equation}
    As substitutability implies $C_f\left( B_f^\mu\cup B_f^{\mu'} \right)\cap B_f^\mu \subseteq C_f(B_f^\mu)$, by \eqref{yes 2} it follows that there is $w \in B_f^{\mu'}\setminus B_f^\mu$ such that $w \in C_f\left( B_f^\mu\cup B_f^{\mu'} \right)$. Since $\mu(f) \subseteq  B_f^\mu$, substitutability implies then that $w \in C_f(\mu(f)\cup \{w\})$. Therefore, $f \in F_w^\mu$.
    
    \medskip
    
    \noindent \textbf{Claim: $\boldsymbol{F_w^\mu \subseteq F_w^{\mu'}}$.} For each $f' \in F$, let $S_{f'}'=\mu(f') \cup \mu'(f')$ and $S_{f'}=\mu'(f')$. Then, $S_{f'} \subseteq S_{f'}'$. Moreover, $\mu'(f')=C_{f'}(S_{f'})$ for each $f' \in F$ by the individual rationality of $\mu'$, and $\mu(f')=C_{f'}(S_{f'}')$ for each $f' \in F$ because $\mu \succeq_F \mu'$. Applying Lemma \ref{lema Blair} (ii) with $\mu'$ in the role of $\mu$ and vice versa, $F_w^\mu \subseteq F_w^{\mu'}$. This proves the Claim.
    
    \medskip

    \noindent Since $w \notin B_f^\mu$, \begin{equation}\label{la final}
    f \notin C_w(F_w^\mu).    
    \end{equation}
    As $f \in F_w^\mu$, by \eqref{la final} together with the Claim and substitutability it follows that $f \notin C_w(F_w^{\mu'})$. This last fact and  $w \in B_f^{\mu'}$ imply that $w\in \mu'(f)$. Thus, $f\in\mu'(w)$. It follows from quasi-stability and  Lemma \ref{lema Blair} (i) that $$f \in \mu'(w)\subseteq C_w(\mu'(w)\cup F^{\mu'}_w)=C_w(F^{\mu'}_w).$$ Therefore,  $f\in C_w(F^{\mu'}_w)$, contradicting \eqref{la final}. Hence, $\mathcal{T}^F[\mu] \succeq_F \mathcal{T}^F[\mu']$, as desired.

\bigskip

\noindent (iii) 
\noindent ($\Longrightarrow$) Assume that $\mathcal{T}^F[\mu] =\mu$. Let $f\in F. $  To check that $\mu \in \mathcal{S}$ it is sufficient to see that $B_f^\mu=\mu(f)$. By \eqref{b tilde de f}, $\mu(f) \subseteq B_f^\mu$. Let $w \in B_f^\mu$. Then, $f \in C_w(F^\mu_w) \subseteq F_w^\mu$ and thus $w \in C_f(\mu(f)\cup \{w\}).$ Moreover, as $C_f(B_f^\mu)=\mu(f)$, by \eqref{propiedad de choice} we have  
\begin{equation*}\label{Blair 6}
 C_f(\mu(f)\cup \{w\})=C_f\left( C_f(B_f^\mu)\cup \{w\}\right)
=C_f\left(B_f^\mu\cup \{w\}\right)=C_f\left(B_f^\mu \right)=\mu(f).
\end{equation*}
Hence, $w\in \mu(f)$ and $B_f^\mu \subseteq \mu(f)$. Therefore, $B_f^\mu=\mu(f)$ and $\mu \in \mathcal{S}$.

\noindent ($\Longleftarrow$) Assume $\mathcal{T}^F[\mu] \neq \mu$. Then, there is $f \in F$ such that $C_f(B_f^\mu) \neq \mu(f)$. Therefore, there is $w \in C_f(B_f^\mu)$ with $w \notin \mu(f)$. Since $\mu(f) \subseteq B_f^\mu$, by substitutability, \begin{equation}\label{Blair 4}
    w \in C_f(\mu(f)\cup \{w\}).
\end{equation}
Moreover, since $w \in B_f^\mu$, $f \in C_w(F_w^\mu)$. Since $\mu(w) \subseteq F_w^\mu$ by Lemma \ref{lema Blair} (i),  substitutability implies \begin{equation}\label{Blair 5}
    f \in C_w(\mu(w) \cup \{f\}).
\end{equation}
Therefore, as $w \notin \mu(f)$, by \eqref{Blair 4} and \eqref{Blair 5} it follows that $(f,w)$ blocks $\mu$. Hence, $\mu \notin \mathcal{S}$.  \hfill $\square$

\subsection{Proof of Theorem \ref{theorem join and meet many-to-many}}\label{proof of theorem join and meet many-to-many}

\noindent Let $\mu,\mu' \in \mathcal{S}$. First, let us see (i). For short, let $\lambda\equiv \lambda_{\mu, \mu'}$.  By Theorem \ref{prop facts worker-quasi many} (i), $\lambda \in \mathcal{Q}^W$  and, furthermore, $\lambda$ is the join of $\mu$ and $\mu'$ with respect to $\succeq_F$ in $\mathcal{Q}^W$. Define $$\Lambda=\{\nu \in \mathcal{Q}^W : \nu \succeq_F \lambda\}.$$

\medskip

    \noindent \textbf{Claim: $\boldsymbol{\mathcal{T}^F[\Lambda]\subseteq \Lambda}$}. Let $\nu \in \Lambda$. Then, $\nu \succeq_F \lambda\succeq_F \mu$ and $\nu \succeq_F \lambda \succeq_F \mu'$. By Theorem  \ref{results operator worker many-to-many} (ii) and (iii), $\mathcal{T}^F(\nu) \succeq_F \mathcal{T}^F(\mu)=\mu$ and $\mathcal{T}^F(\nu) \succeq_F \mathcal{T}^F(\mu')=\mu'$. This implies that   $\mathcal{T}^F(\nu)$ is an upper bound of $\mu$ and $\mu'$ in $\mathcal{Q}^W$. By definition of join, $\mathcal{T}^F(\nu)\succeq_F \lambda$ and thus $\mathcal{T}^F(\nu) \in \Lambda$. This proves the claim. 

\medskip
    
    Next, consider the restriction of $\mathcal{T}^F$ to $\Lambda$. By the previous claim, $\mathcal{T}^F|_{\Lambda}: \Lambda \longrightarrow \Lambda$.  
    Since $\mathcal{T}^F$ is isotone by Theorem \ref{results operator worker many-to-many} (ii) and $\lambda$ is the $\succeq_F$-smallest element in $\Lambda$, $\mathcal{F}^F(\lambda)$ is the $\succeq_F$-smallest fixed point larger than $\mu$ and $\mu'$. Otherwise, if $\delta$ is an upper bound of $\mu$ and $\mu'$, $\mathcal{F}^F(\lambda) \succ_F \delta$, and $\delta$ is a fixed point of $\mathcal{T}^F$; as $\delta \succeq_F \lambda$ isotonicity implies $\delta \succeq_F \mathcal{F}^F(\lambda)$, a contradiction. Finally, since $\mathcal{F}^F(\lambda) \in \mathcal{S}$ by Theorem  \ref{results operator worker many-to-many} (iii), we get $\mathcal{F}^F(\lambda)= \mu \ \curlyvee_F \mu'$, as desired. 

    To see (ii), first apply \emph{mutatis mutandis} the same reasoning as before, but this time to operator $\mathcal{T}^W$ using results dual to those in Theorem \ref{results operator worker many-to-many}, to obtain $ \mu  \curlyvee_W 
  \mu'=\mathcal{F}^W(\gamma_{\mu, \mu'})$. Next, by Theorem 4.5 in \cite{blair1988lattice}, $\succeq_F$ and $\succeq_W$ are dual orders in $\mathcal{S}$. Therefore, $\mu  \curlywedge_F 
  \mu'=\mu  \curlyvee_W 
  \mu'$ and the result follows. \hfill $\square$

\end{document}